\documentclass[lettersize,journal]{IEEEtran}
\IEEEoverridecommandlockouts
\pdfoutput=1
\usepackage{amsmath,amsfonts}
\usepackage{algorithmic}
\usepackage{algorithm}
\usepackage{array}
\usepackage[caption=false,font=normalsize,labelfont=sf,textfont=sf]{subfig}
\usepackage{textcomp}
\usepackage{stfloats}
\usepackage{url}
\usepackage{verbatim}
\usepackage{graphicx}
\usepackage{cite}
\usepackage{amsmath,amsthm}
\newtheorem{theorem}{Theorem}
\begin{document}
\bibliographystyle{IEEEtran}
\title{On the Performance Trade-off of Distributed  Integrated Sensing and Communication Networks}

\author{Xuran Li, Shuaishuai Guo,~\IEEEmembership{Senior Member, IEEE,}  Tuo Li, Xiaofeng Zou, Dengwang Li

\thanks{Xuran Li and Dengwang Li are  with Shandong Key Laboratory of Medical Physics and Image Processing, School of Physics and Electronics, Shandong Normal University, Jinan 250061, China; (e-mail:sdnulxr@sdnu.edu.cn; dengwang@sdnu.edu.cn).}
\thanks{Shuaishuai Guo is with School of Control Science and Engineering, Shandong University, Jinan 250061, China and also  with Shandong Provincial Key Laboratory of Wireless Communication Technologies(e-mail: shuaishuai\textunderscore guo@sdu.edu.cn).}
\thanks{Tuo Li and Xiaofeng Zou are with Shandong Yunhai Guochuang Cloud Computing Equipment Industry Innovation Co., Ltd; (e-mail:lituo@inspur.com; zouxf@inspur.com).}
\thanks{(\emph{Corresponding author: Shuaishuai Guo}.)}
}



\maketitle

\begin{abstract}
In this letter, we analyze the performance trade-off in distributed integrated sensing and communication (ISAC) networks. Specifically, with the aid of stochastic geometry theory, we derive the probability of detection of that of the coverage given user number. Based on the analytical derivations, we provide a quantitative description of the performance limits and the performance trade-off between sensing and communication in a distributed ISAC network under the given transmit power and bandwidth budget. Extensive simulations are conducted and the numerical results validate the accuracy of our derivations.
\end{abstract}



\section{Introduction}

Integrated sensing and communication (ISAC) emerges as a promising solution to enable wireless communication and radar sensing in the same system\cite{LiuFan:2022JSAC}.
By utilizing the same spectrum, hardware platform, and signal processing modules, ISAC can improve spectrum and energy efficiency, thereby solving spectrum congestion issues, while reducing hardware and signaling costs\cite{Zhiqiang:2022JSAC}.
Although the ISAC system has the potential to improve the efficiency of energy, spectrum, and hardware of the whole system, the resource competition (i.e. limited power, spectrum, antennas, etc.) between the sensing  and communication functionalities is an important problem to be solved\cite{ShuaishuaiGuo:Mobicom2022}. Therefore, it is essential to figure out the fundamental performance limits and performance trade-offs between sensing  and communication functionalities in ISAC systems under resource competition\cite{CongDingyan:WCL2022}.

Figuring out the fundamental performance trade-off between sensing and communication is a core problem in ISAC networks with limited wireless resources (i.e. power, spectrum, antennas, etc.). Current communication-sensing performance trade-off analysis in ISAC networks can be classified into two categories. In centralized ISAC networks, such as cellular networks, Kafafy \emph{et al.} considered the performance trade-off of the ISAC cellular system sharing spectrum with a rotating radar \cite{Kafafy:2021ICL}. The performance trade-offs of downlink and uplink in the ISAC cellular system were separately analyzed with the metric communication rate and sensing rate in \cite{Ouyang:2022WCL,ZhangChao:2023Tcom}. The performance trade-off of a multi-point ISAC system was analyzed in \cite{Guoliang:2022WCL}, and the fusion center of this multi-point ISAC system fuses the outputs from multiple ISAC devices to achieve higher detection accuracy.
In distributed ISAC networks, such as vehicular networks \cite{CongDingyan:TITS2022}, the authors analyzed the spectral efficiency of uplink communications to improve the performance of their beamforming design. A convolutional neural network was presented to fuse the ISAC signals and improve the energy efficiency of a distributed Internet of Things (IoT) network \cite{Zexuan:2023IoTJ}. The probability of detection and spectrum efficiency of a full-duplex ISAC scheme were analyzed in unmanned aerial vehicle (UAV) networks \cite{Zhiqiang:2022JSAC}.

It is noteworthy that most of the existing works rely on deterministic assumptions. For example, the number of ISAC users in the networks is assumed to be fixed \cite{Zhiqiang:2022JSAC,CongDingyan:TITS2022}, and the channel state information for communication and the potential position of sensing targets are assumed known \cite{Zexuan:2023IoTJ,Kafafy:2021ICL,Ouyang:2022WCL,ZhangChao:2023Tcom}.
Rather than relying on deterministic assumptions, stochastic geometry investigates the inherent randomness and uncertainties in network deployments. Besides, stochastic geometry facilitates the scalability analysis of networks. It allows for evaluating network performance as the number of network elements scales \cite{Xiaoming:2022TWC}.


In this letter, we propose a  stochastic geometry-based performance analysis framework to analyze distributed ISAC networks. Based on the proposed analytical framework, we derive mathematical closed-form expressions for the probability of detection in the sensing system and the probability of coverage in the communication system. 
The tractable analytical results of the derivations demonstrated the fundamental performance limits of the distributed ISAC network. Through extensive simulations, we validate the correctness and accuracy of our derivations.

\section{System model}
In this section, we introduce the ISAC system model, including the network model, channel model, and performance analysis model.

\subsection{Network Model}

\begin{figure}[ht]
\centering
\includegraphics[width=2.95in]{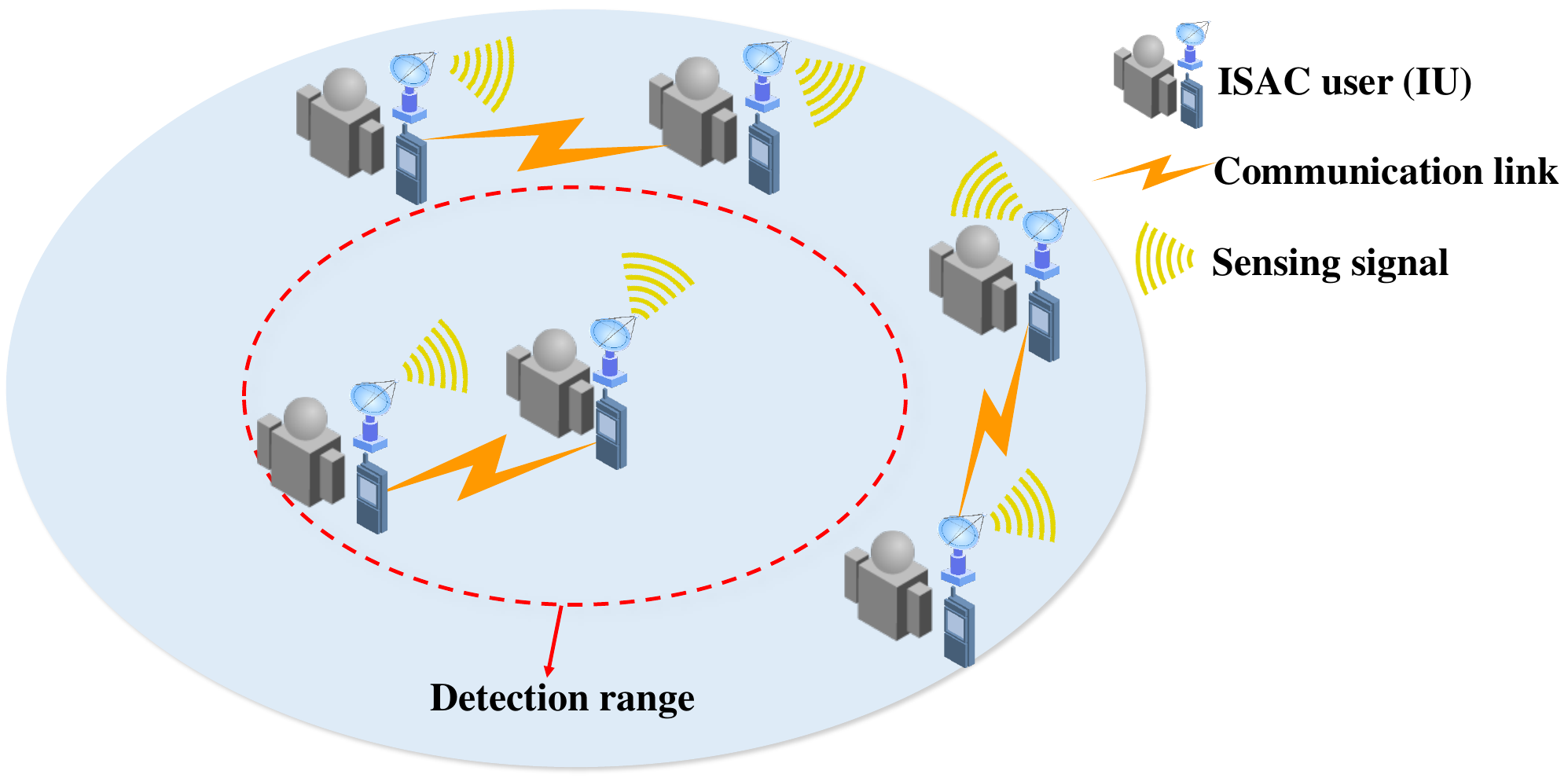}
\caption{An ISAC Network.}
\label{fig:network_model}
\end{figure}

In this paper, we consider an ISAC network, as shown in Fig. \ref{fig:network_model}.
In this network, multiple ISAC users (IU) are randomly distributed according to a homogeneous Poisson point process (PPP) in the circular network area. The IUs are assumed to communicate with each other directly without intervention from the base station. This assumption is feasible for many communication scenarios in that ISAC technology can be used, such as vehicular networks \cite{CongDingyan:TITS2022}, distributed IoT networks\mbox{\cite{Zexuan:2023IoTJ}}, UAV networks \cite{Zhiqiang:2022JSAC}. Without loss of generality, we assume the typical IU locates at the center of the circular area \cite{Xuran:2022JSAC}.


Each IU of the network transmits a sensing signal for target sensing, meanwhile, they transmit information signals for communication between network users.
All the IUs are assumed to share common bandwidth $B$ and transmit power $p_t$, their transmitted sensing signals and information signals interfere with each other \cite{Kafafy:2021ICL}.
In some ISAC networks, IUs may share the time resource of the ISAC system, while the performance trade-off of communication and sensing is a simple linear relationship in time domain \cite{LiuFan:2022JSAC}. Therefore, we focus on the performance trade-off of communication and sensing when the total bandwidth $B$ and transmit power $p_t$ are shared in this paper.

\subsection{Channel Model}
\label{subsec:channel_model}
Without loss of generality, we assume the wireless transmission channel experience Rayleigh fading and path loss. When the sensing target is with distance $D$ from the typical IU, the received target echo signal at the typical IU is expressed as \cite{Ren:2019WCL}
\begin{equation}
P_{e c h}=\frac{p_s G_t G_r \lambda_w^2 \sigma}{(4 \pi)^3 D^{2 \alpha}},
\label{eq:Pech}
\end{equation}
where $p_s$ is the transmitting power of the sensing signal; $G_t$ and $G_r$ are the transmitting antenna gain and receiving antenna gain separately; $\lambda_w$ is the wavelength of the sensing signal; $\alpha$ is the path loss factor, and $\sigma$ is the cross-section of the target.

We describe the cross-section of the target $\sigma$ as a random variable by the Swerling type-1 model \cite{Fang:2020WCL}, then we have the probability density function (PDF) of $\sigma$ given by
\begin{equation}
f (\sigma)=\frac{1}{\bar{\sigma}} e^{-\frac{\sigma}{\bar{\sigma}}}, \sigma \geq 0
\label{eq:sigma}
\end{equation}
where $\bar{\sigma}$ is the average cross-section of the target. This expression of $\sigma$ also follows Chi-square distribution with a degree of freedom $1$. This model is feasible for a relatively slow-moving target (like a human or vehicle).

When the sensing function is active, the sensing device conduct detection decisions based on the target echo power $P_{e c h}$. With the received target echo signal at the typical user $P_{e c h}$, we can derive the signal interference and noise ratio (SINR) of the sensing signal at the typical user as \cite{Fang:2020WCL,Guoliang:2022WCL}
\begin{equation}
\mathrm{SINR_s}=\frac{P_{e c h}}{I_s+N_0},
\end{equation}
where $N_0=K_0 T_0 B_s L$ is the power of additive white Gaussian noise; $K_0$ is the Boltzmann constant; $T_0$ is the standard temperature; $B_s$ is the noise bandwidth and $L$ is the system loss factor; $I_s$ is the aggregated interference from other IUs in the network.

In specific, the aggregated interference $I_s$ from other IUs is given by $I_s=\sum_{j=1}^n A_e S h_j l_j ^{-\alpha}$\cite{Fang:2020WCL,Ren:2019WCL},
where $A_e=\frac{G_r \lambda_w^2}{4 \pi}$ is the effective aperture of receiving antenna; $S=\frac{p_s G_t}{4 \pi}$ is the power density at a unit distance from an interferer \cite{Fang:2020WCL}; $n$ is the expected number of interferer IU devices; $l_j$ is the distance between the $j$th interferer IU device and the typical IU; $\alpha$ is the path loss factor; $h_j$ is the Rayleigh fading variable following an exponential distribution with mean $1$.

To evaluate the performance of radar targets, the metric ergodic radar estimation information rate for radar echoes is proposed in \cite{Chiriyath:2016TSP}.
This metric is analogous to the data information rate of the communications system. The relationship between the ergodic radar estimation information rate $R_{B}$ and the SINR is given by $R_{B} = \frac{\delta}{2 T} \log _2 ( 1 + 2 T  B_s \cdot \mathrm{SINR_s} ) $ \mbox{\cite{ZhangChao:2023Tcom}}, where $T$ is the radar pulse duration, and $\delta$ is the radar's duty cycle. From the expression of $R_{B}$, we find this metric is the calculated estimation rate of the parameters (range, bandwidth, power, etc.), and a higher $R_{B}$ indicates better performance for radar detection.

With a similar approach, we can derive the SINR of the information signal at the desired transmission receiver by \cite{JeongSeon:2023TVT}
\begin{equation}
\mathrm{SINR_c} = \frac{p_c G_t G_r {h_0}{{r_0}^{ - \alpha }}}{N_c + {I_c}},
\label{eq:sinrC}
\end{equation}
where $p_c$ is the transmitting power of information signal of the typical IU device; $r_0$ is the distance between the typical IU and the desired information receiver; $N_c$ is the additive white Gaussian noise; ${I_c}$ is the aggregated interference from other IUs in the network. In specific, ${I_c} =\sum\limits_{i \in \Phi /{t_0}}p_c G_t G_r {{h_i}r_i^{ - \alpha }}$ \cite{JeongSeon:2023TVT}, $\Phi$ is the set of IU devices that uniformly distributed following PPP, $p_c$ is the transmitting power of interferer information signals, $r_i$ is the distance between the $i$th interferer IU and the typical IU.

With the bandwidth of communication $B_c$ and the SINR of the information signal in equation (\ref{eq:sinrC}), the achievable transmission capacity is $C_{B}= B_c \log _2 (1+ \mathrm{SINR_c})$ \cite{Xiaoqian:TMC2022}.

\subsection{Performance Analysis Model}
\label{subsec:perfor_model}
Since the IU transmits the sensing signal and information signal simultaneously, we seek to develop a unified and efficient framework to analyze the sensing performance and communication performance.
Then we introduce the well-known performance metrics: the probability of detection $P_D$ and the probability of coverage $P_{C}$.
The probability of detection $P_D (\zeta_s)$ is the probability that the $R_{B}$ at the typical IU exceeds a predefined threshold value $\zeta_s$.
The probability of coverage $P_C(\zeta_c)$ is the probability that the $C_{B}$ at the desired IU receiver is above a predefined threshold value $\zeta_c$.

\section{Performance Analysis}
\label{sec:analysis}
In this section, we focus on analyzing the sensing performance and communication performance to figure out the performance trade-off in the distributed ISAC network.

\subsection{Sensing Performance Analysis}
We first analyze the sensing performance of the typical IU in the distributed ISAC network. The analytical expression of probability of detection $P_D$ is derived as follows.

\begin{theorem}
\label{theorem:Pinf_s}
The probability of detection $P_D$ in the ISAC network can be expressed by
\begin{equation}
\label{PD}
{P_D(\zeta_s)=\exp \left(-\frac{\eta D^{2 \alpha} N_0}{\bar{\sigma}}\right)   \left[ \frac{2 }{R^2} {  \int_0^R  (\frac{\bar{\sigma} r^{ 1 + \alpha }}{\bar{\sigma} r^{ \alpha } + 4 \pi \gamma_s D^{2 \alpha}  } )dr } \right]^n},
\end{equation}
where $\eta=\frac{\gamma_s (4 \pi)^3}{p_s G_t G_r \lambda_w^2}$ and $\gamma_s= \frac{ 2^{ \frac{2 T \zeta_s}{\delta} } -1 }{2 T B_s} $.
\end{theorem}
\begin{proof}

From the definition of probability of detection $P_D$ given in Section~\ref{subsec:perfor_model} and the expression of $R_{B}$, we have \cite{Fang:2020WCL},
\begin{equation}
\begin{aligned}
P_D(\zeta_s)=  P[R_{B}>\zeta_s]=P\left[\mathrm{SINR_s} >  \frac{ 2^{ \frac{2 T  \zeta_s}{\delta} } -1 }{2 T B_s} \right].
\label{eq:PDorigin}
\end{aligned}
\end{equation}
For simplicity, we denote $\gamma_s= \frac{ 2^{ \frac{2 T  \zeta_s}{\delta} } -1 }{2 T B_s} $. Then we substitute the expression of $\mathrm{SINR_s}$ in \eqref{eq:sigma} and the expression of $P_{e c h}$ in \eqref{eq:Pech} into \eqref{eq:PDorigin}, the above equation is expressed by the following,
\begin{equation}
\begin{aligned}
P_D(\zeta_s)=P[\mathrm{SINR_s}>\gamma_s]=P\left[\sigma>\frac{\gamma_s(4 \pi)^3 D^{2 \alpha}}{p_s G_t G_r \lambda_w^2}\left(I_s+N_0\right)\right].
\label{eq:PD}
\end{aligned}
\end{equation}

As introduced in Sections~\ref{subsec:channel_model},
the cross-section of the target ${\sigma}$ is a random
variable following the Swerling type-1 model. With the PDF given by \eqref{eq:sigma}, we obtain
\begin{equation}
P_D(\zeta_s)=\exp \left(-\frac{\eta D^{2 \alpha} N_0}{\bar{\sigma}}\right) E_{I_s}\left[\exp \left(-\frac{\eta D^{2 \alpha}}{\bar{\sigma}} I_s \right)\right],
\end{equation}
where $\eta=\frac{\gamma_s (4 \pi)^3}{p_s G_t G_r \lambda_w^2}$ for simplicity.

Next, we calculate term ${E_{I_s}}[\exp (-\frac{\eta D^{2 \alpha}}{\bar{\sigma}} I_s)]$ as \mbox{\cite{Xiaoming:2022TWC,Xuran:2022JSAC}},

\begin{equation}
\begin{aligned}
\label{eq:PC_E_It}
& {E_{I_s}}\left[\exp (-\frac{\eta D^{2 \alpha}}{\bar{\sigma}} I_s)\right]\\ \overset{(a)} =  &  {E_{\Phi ,\{ {h_i}\} }}\left[\exp ( -\frac{\eta D^{2 \alpha}}{\bar{\sigma}} \sum_{j=1}^n A_e S h_j l_j ^{-\alpha} )\right] \\
\overset{(b)} =  & {E_{\{{l_j}\} ,\{{h_j}\} }}\left[\exp ( -\frac{\eta D^{2 \alpha} A_e S}{\bar{\sigma}} \sum_{j=1}^n  h_j l_j ^{-\alpha} )\right] \\
\overset{(c)} =  & {E_{\{{l_j}\} ,\{{h_j}\} }}\left[ \prod\limits_{j = 1}^{n}  \exp ( -\frac{4\pi \gamma_s D^{2 \alpha}}{\bar{\sigma}} h_j l_j ^{-\alpha})\right] \\
\overset{(d)} = & \left[{E_{l, h}} \left(\exp ( - h l^{-\alpha} \frac{4\pi \gamma_s D^{2 \alpha}}{\bar{\sigma}}  ) \right)\right]^n  \\
\overset{(e)} = & \left[{E_{l}} \left( {\frac{\bar{\sigma} l^{ \alpha }}{\bar{\sigma} l^{ \alpha } + 4\pi \gamma_s D^{2 \alpha} }} \right) \right]^n \\
\overset{(f)} =  & \left[ {  \int_0^R  (\frac{\bar{\sigma} l^{ \alpha }}{\bar{\sigma} l^{ \alpha } + 4\pi \gamma_s D^{2 \alpha} } ) {f_l}(l)dr } \right]^n.
\end{aligned}
\end{equation}
where $(a)$ follows by substituting the expression of $I_s$, and $l$ is the distance between the IU transmitter and the sensing target; $(b)$ is derived due to the property of Poisson Point Process (PPP) of $\Phi$; $(c)$ is found by the summation properties of exponential function; $(d)$ is derived from the assumption that Rayleigh fading factor of each channel follows exponentially independent identical distribution, and the distance from each IU to target follows same independent identical distribution; $(e)$ is derived from the property of moment generating function of the exponential variable; $(f)$ is derived due to the definition of the expectation.

Since the typical IU transmitter is located at the center of a circular network area and each IU in this network follows uniformly i.i.d., the probability density function of $l$ can be given by ${f_l}(l) = \frac{2\pi l}{\pi {R^2}} = \frac{2l}{R^2}$ \cite{Yujie:2023TCOM}, where $R$ is the radius of the circular network area and $~0 < l \le R$.

Substituting ${f_l}(l)$  into ~\eqref{eq:PC_E_It}, we can calculate
\begin{equation}
\label{eq:PC_E_It3}
{E_{I_s}}\left[\exp (-\frac{\eta D^{2 \alpha}}{\bar{\sigma}} I_s)\right]
=  \left[ \frac{2 }{R^2} {  \int_0^R  (\frac{\bar{\sigma} l^{ 1 + \alpha }}{\bar{\sigma} l^{ \alpha } + 4\pi \gamma_s D^{2 \alpha} } )dr } \right]^n ,
\end{equation}

Then we derive the complete expression of the probability of detection $P_D$ as (\ref{PD}).
\end{proof}
\subsection{Communication Performance Analysis}
\label{subsec:commmunication}

Then we analyze the communication performance of the typical IU and obtain the probability of coverage $P_C$ of the distributed ISAC network as follows.

\begin{theorem}
\label{theorem:PT_f}
The probability of coverage $P_C$ in the ISAC network can be expressed by
\begin{equation}
 {\begin{aligned}
\label{eq:PT_f}
& P_C(\zeta_c) =
\\& \frac{2^m}{R^{2m}} \displaystyle{\int_0^R {\exp ( - {\gamma_p}{r_0}^\alpha N_c)
{\left( {\frac{2}{{{R^2}}}\int_0^R {\frac{{{r^{1 + \alpha }}}}{{{r^\alpha }{\rm{ + }}\gamma_c
 {r_0}^\alpha }}}dr} \right)^{m - 1}}} d{r_0}},
\end{aligned}}
\end{equation}
where ${\gamma_p}= \gamma_c /(p_c G_t G_r)$,  $\gamma_c = 2^{ \zeta_c / B_c } -1$ and $m$ is the expectation of the number of legitimate IUs in the communication area.
\end{theorem}

\begin{proof}
Let $r$ denote the distance between the typical IU receiver and an IU transmitter. Since the receiver is located at the center of a circular area and each transmitter follows an independent and identical uniform distribution, the probability density function of $r$ is given by ${f_r}(r) = \frac{2\pi r}{\pi {R^2}} = \frac{2r}{R^2}$\cite{Yujie:2023TCOM} .

From the definition of the probability of coverage $P_C$ in Section~\ref{subsec:perfor_model}, we obtain ${P_C(\zeta_c)}$ as\mbox{\cite{JeongSeon:2023TVT}},
\begin{equation}
\begin{split}
\label{eq:PC_fh}
 P_C(\zeta_c)= &  P[C_{B}>\zeta_c]\\
 \overset{(a)} = & P[\mathrm{SINR_c} >  2^{ \zeta_c / B_c } -1] \\
 \overset{(b)} = &  \int_0^R {P\left[ {\frac{{p_c G_t G_r{h_0}{{r_0}^{ - \alpha }}}}{{N_c + {I_c}}} \ge \gamma_c |{r_0}} \right] {f_r}({r_0})d{r_0}} \\
 \overset{(c)} =  &  \int_0^R {P[{h_0} \ge {\gamma_p} {{r_0}^\alpha }(N_c + {I_c} )|{r_0}] {2{r_0}}{R^{-2}}d{r_0}},
\end{split}
\end{equation}
where $(a)$ holds by substituting the expression of $C_{B}$, $(b)$ is found by substituting the expression of $\mathrm{SINR_c}$ in  (\ref{eq:sinrC}) and denoting $\gamma_c = 2^{ \zeta_c / B_c } -1$, $(c)$ is derived by substituting the ${f_r}(r)$ and denoting ${\gamma_p}= \gamma_c /(p_c G_t G_r)$.

Since $h$ is a random variable following an exponential distribution with mean $1$, ${P[{h_0} \ge {\gamma_p} {r_0^\alpha }({\sigma ^2} + {I_c})|{r_0}]}$ in  (\ref{eq:PC_fh}) is represented as \cite{Yujie:2023TCOM} ,
\begin{equation}
\begin{aligned}
\label{eq:PC_h}
 P [{h_0} \ge {\gamma_p} {r_0^\alpha }(N_c + {I_c})|{r_0}]
 = {e^{ -  {\gamma_p}{{r_0}^\alpha } N_c}}  {E_{{I_c}}}[\exp ( - {\gamma_p} {{r_0}^\alpha }{I_c})] .\\
\end{aligned}
\end{equation}

Let $m$ denote the expected number of legitimate transmitters, we can derive the expression of the term {${E_{{I_c}}}\left[\exp ( - \gamma_p {{r_0}^\alpha }{I_c})\right]$} by \cite{Xiaoming:2022TWC,Xuran:2022JSAC},
\begin{equation}
\begin{aligned}
\label{eq:PC_E_Ic}
{E_{{I_c}}}& \left[\exp ( - \gamma_p {{r_0}^\alpha }{I_c})\right] \\
\overset{(a)} = & {E_{\Phi ,\{ {h_i}\} }}[\exp ( -  \gamma_p {r_0}^\alpha \sum\limits_{i \in \Phi /{t_0}} p_c G_t G_r {{h_i}r_i^{ - \alpha }} )] \\
\overset{(b)} =  &  {E_{\{{r_i}\} ,\{{h_i}\} }}[ \prod\limits_{i = 1}^{m-1}  \exp ( - \gamma_c {{r_0}^\alpha} {{h_i}r_i^{ - \alpha }} )] \\
\overset{(c)} = &  \left[{E_{r, h}} (\exp ( - \gamma_c {{r_0}^\alpha} {{h}{r}^{ - \alpha }}) )\right]^{m-1} \\
\overset{(d)} =  & \left[{E_{r}} \left( {\frac{1 }{1 + \gamma_c {r_0}^\alpha r^{ - \alpha }}} \right) \right]^{m-1} \\
\overset{(e)} =  & \left[ \frac{2 }{R^2} {  \int_0^R  ({{\frac{r^{1+\alpha} }{r^{\alpha} + \gamma_c  {r_0}^\alpha }} ) dr} } \right]^{m-1},
\end{aligned}
\end{equation}
where $(a)$ is derived by substituting the expression of $I_c$, $(b)$ is found by the property of exponential function, $(c)$ holds due to the assumption that Rayleigh fading factors are exponentially independent and identically distributed, $(d)$ is derived from the property of the moment generating function of an exponential variable, $(e)$ holds by substituting the expression of ${f_r}(r)$.

Substituting  (\ref{eq:PC_E_Ic}) into the corresponding part of  (\ref{eq:PC_h}), we obtain the expression of $P[{h_0}\ge{\gamma_p}{r^\alpha }(N_c + {I_c})]$ as
\begin{equation}
 {\begin{aligned}
\label{eq:PC_E_hfN}
P&[{h_0}\ge{\gamma_p}{r^\alpha }(N_c + {I_c})] \\
& = \exp ( - {\gamma_p}{r_0}^\alpha N_0)
{\left( {\frac{2}{{{R^2}}}\int_0^R {\frac{{{r^{1 + \alpha }}}}{{{r^\alpha }{\rm{ + }}\gamma_c
 {r_0}^\alpha }}}dr} \right)^{m - 1}}.
\end{aligned}}
\end{equation}

After plugging  (\ref{eq:PC_E_hfN}) into  (\ref{eq:PC_fh}), the expression of $P_C$ can be given as (\ref{eq:PT_f}).
\end{proof}

\begin{table}[!t]
\caption{Simulation Parameters\label{tab:table1}}
\begin{tabular}{c|c}
\hline \hline Sensing Parameters & Values  \\
\hline Wavelength of sensing signal $\lambda_w$   & $0.0833~\mathrm{m} $ \cite{Fang:2020WCL}\\
Transmitted power $p_{s}$  & 0 to $30~\mathrm{dBm}$ \cite{Zhiqiang:2022JSAC}\\
Antenna gain $G_t,G_r$ & $\mathrm{10~dBi}$ \\
Detection distance $D$ & 5 to $50 \mathrm{~m}$ \cite{Fang:2020WCL}\\
Standard temperature $T_o$  & $290 \mathrm{~K}$ \cite{Fang:2020WCL}\\
Radar bandwidth $B_s$ & 0 to $20 \mathrm{~MHz}$ \cite{Xiaoqian:TMC2022}\\
System loss factor $L$&  $10 \mathrm{~dB}$ \cite{Fang:2020WCL}  \\
Boltzmann constant $K_0$ & $1.38 \times 10^{-23}$ $\mathrm{~w/s}$  \cite{Fang:2020WCL}\\
 Average cross section of target $\bar{\sigma}$
 & $10 \mathrm{~dBsm}$ \cite{Fang:2020WCL} \\
 Radar pulse duration $T$  & 1\textmu s \cite{ZhangChao:2023Tcom} \\
 Radar's duty cycle $\delta$  & 0.01 \cite{ZhangChao:2023Tcom}\\
Expected number of IU interferers $n$ & 10 to 50 \\
\hline Communication Parameters & Values \\
\hline
Radius of communication area $R$ & $0.5\mathrm{~km}$  \cite{Fang:2020WCL}\\
Transmitted power $p_c$  & 0 to 30 $\mathrm{~dBm}$ \cite{Zhiqiang:2022JSAC}\\
Path loss exponent $\alpha$ & 2 to 5 \\
Communication bandwidth $B_c$  & 0 to $20
\mathrm{~MHz}$ \cite{Xiaoqian:TMC2022} \\
Expected number of IU transmitters $m$ &  10 to 50 \\
\hline \hline
\end{tabular}
\end{table}

\subsection{ Performance Trade-off Analysis }
 With the performance metrics of sensing and communication functions in the ISAC system, we can analyze the performance trade-off of sensing and communication. Firstly, we analyze the impact of transmit power allocation on the performance trade-off of sensing and communication. Due to the limited total transmit power budget $p_t$, the $P_{D}$ and $P_{C}$ could not reach the maximal value simultaneously. The power budget condition that caused the performance trade-off is given by \cite{LiuFan:2022JSAC},
\begin{equation}
 p_s+p_c \leq p_t.
\end{equation}

Another commonly used strategy of ISAC is to allocate the sensing and communication waveforms to different frequency bands. When the frequency-division ISAC is applied, there is a portion of the total bandwidth used for sending the sensing signal, and the other portion of the total bandwidth used for sending the information signal. We can express the bandwidth budget condition that caused the performance trade-off as \cite{Xiaoqian:TMC2022},
\begin{equation}
B_s+B_c \leq B.
\end{equation}
where $B_s$ and $B_c$ are the bandwidth of sensing and communication, separately.

 \section{Simulations and Discussions}

In this section, analytical results are provided and
extensive Monte-Carlo simulations are conducted to validate our analysis.
The simulation parameters are summarized in Table \ref{tab:table1}.

\begin{figure}[htbp]
\begin{minipage}[t]{0.49\linewidth}
\centering
\includegraphics[width=1.83 in]{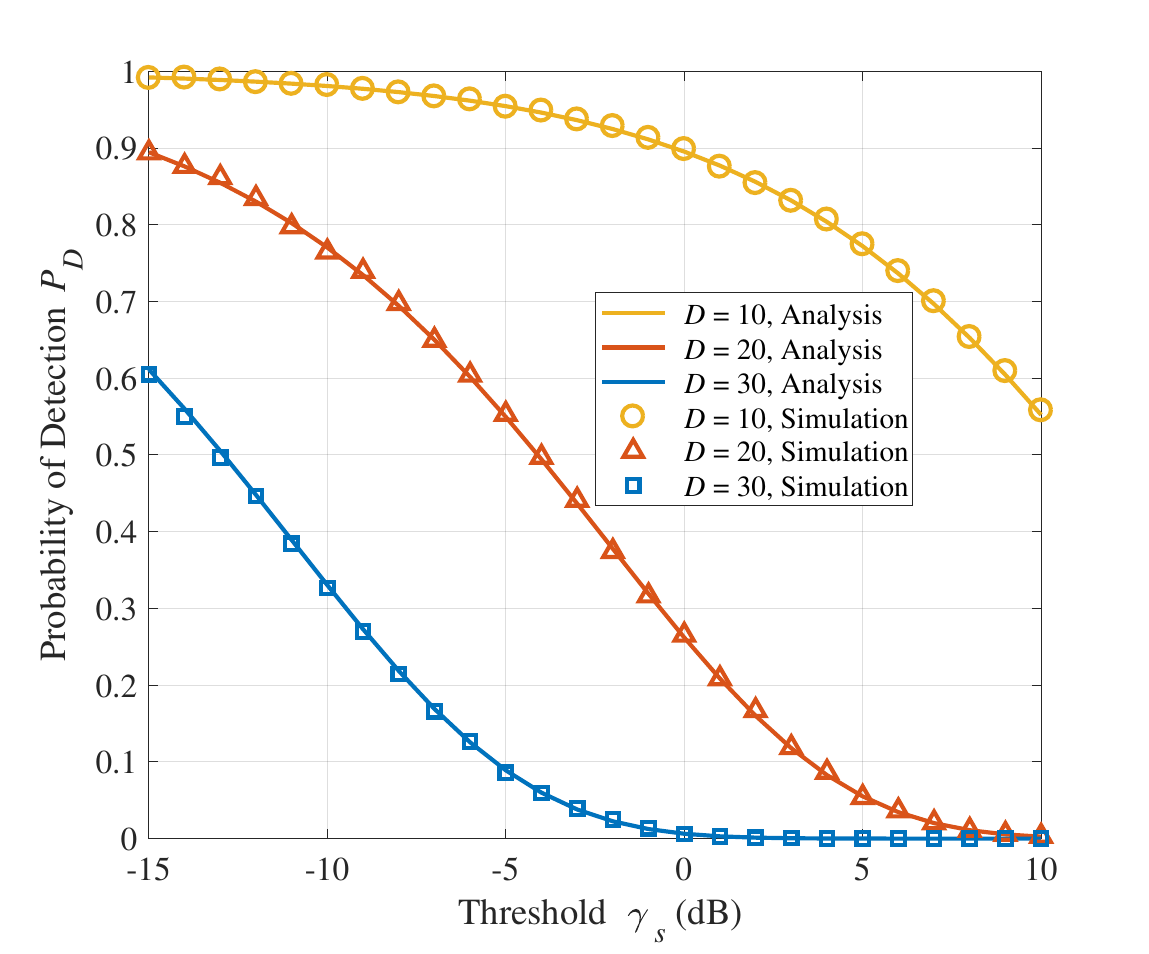}
\caption{Probability of detection $P_D$ varies with threshold value $\gamma_s$.}
\label{fig:reultPDDis}
\end{minipage}
\begin{minipage}[t]{0.49\linewidth}
\centering
\includegraphics[width=1.83 in]{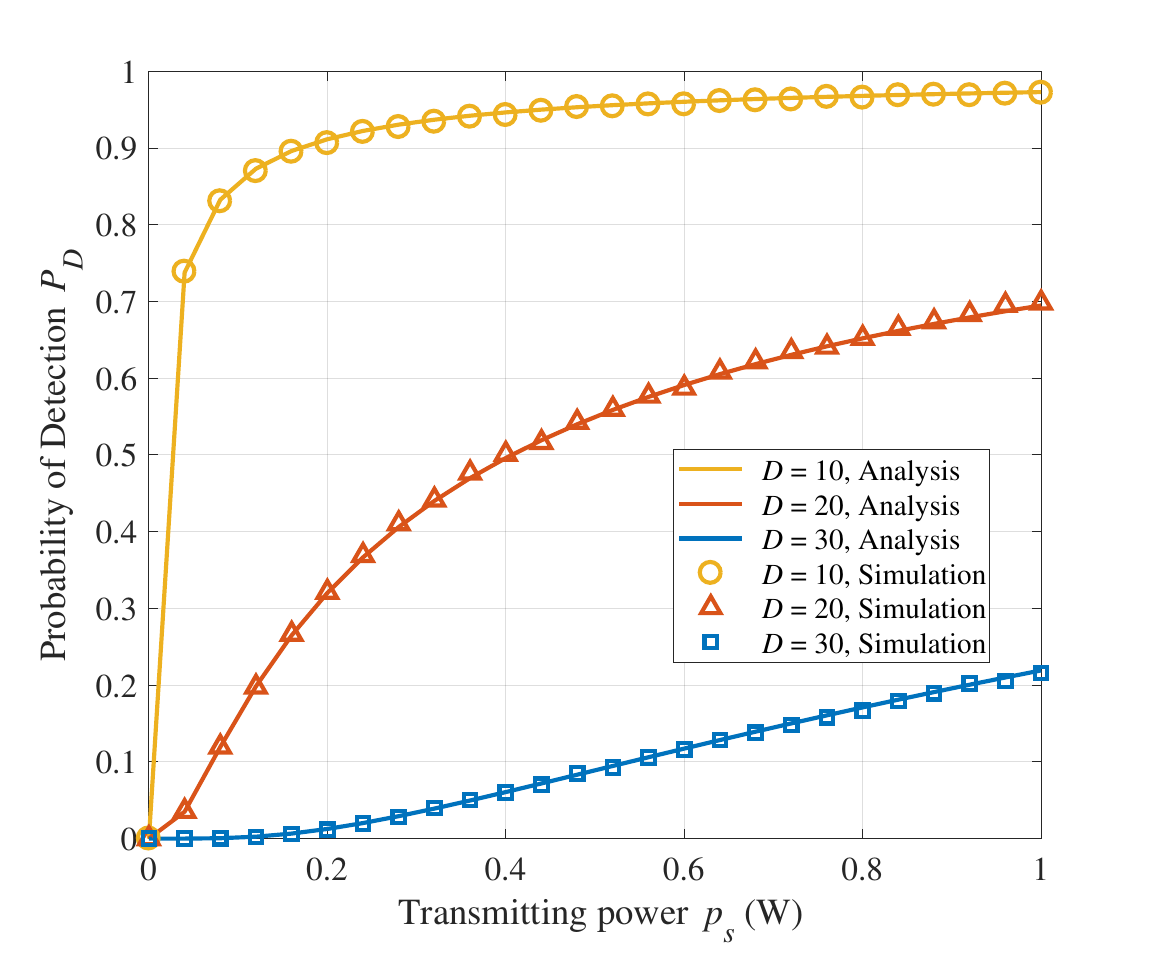}
\caption{Probability of detection $P_D$ varies with transmit power $p_s$.}
\label{fig:reultPDPs}
\end{minipage}
\end{figure}

First, we demonstrate the results of the probability of detection $P_D$ of the desired IU varies with threshold value $\gamma_s$ in Fig. \ref{fig:reultPDDis}. In this figure, different curves represent different detection distances $D$.
When the threshold value increases from a very low level (indicating the capability of IU devices extracting useful signals becoming weak), the probability of detection $P_D$ decreases, and the desired IU is less likely to successfully detect the sensing target.
When the threshold value increases to a high level, the desired IU is likely to successfully detect the sensing target if the sensing target is very near to the desired IU (e.g., less than 10 meters).

Next, we present the results of the probability of detection $P_D$ varies with the transmitting power $p_s$ of the desired IU in Fig.~\ref{fig:reultPDPs}.
When the transmitting power $p_s$ of the desired IU increases, the probability of detection $P_D$ increases. For example, when the detection distance $D$ is 20 meters and the transmitting power $p_s$ is 0.2 W, the probability of detection $P_D$ is 0.3245. When the detection distance $D$ remains to be 20 meters and $p_s$ becomes 0.6 W, the $P_D$ increases to 0.5926.

\begin{figure}[htbp]
\begin{minipage}[t]{0.48\linewidth}
 \centering
 \includegraphics[width=1.83 in]{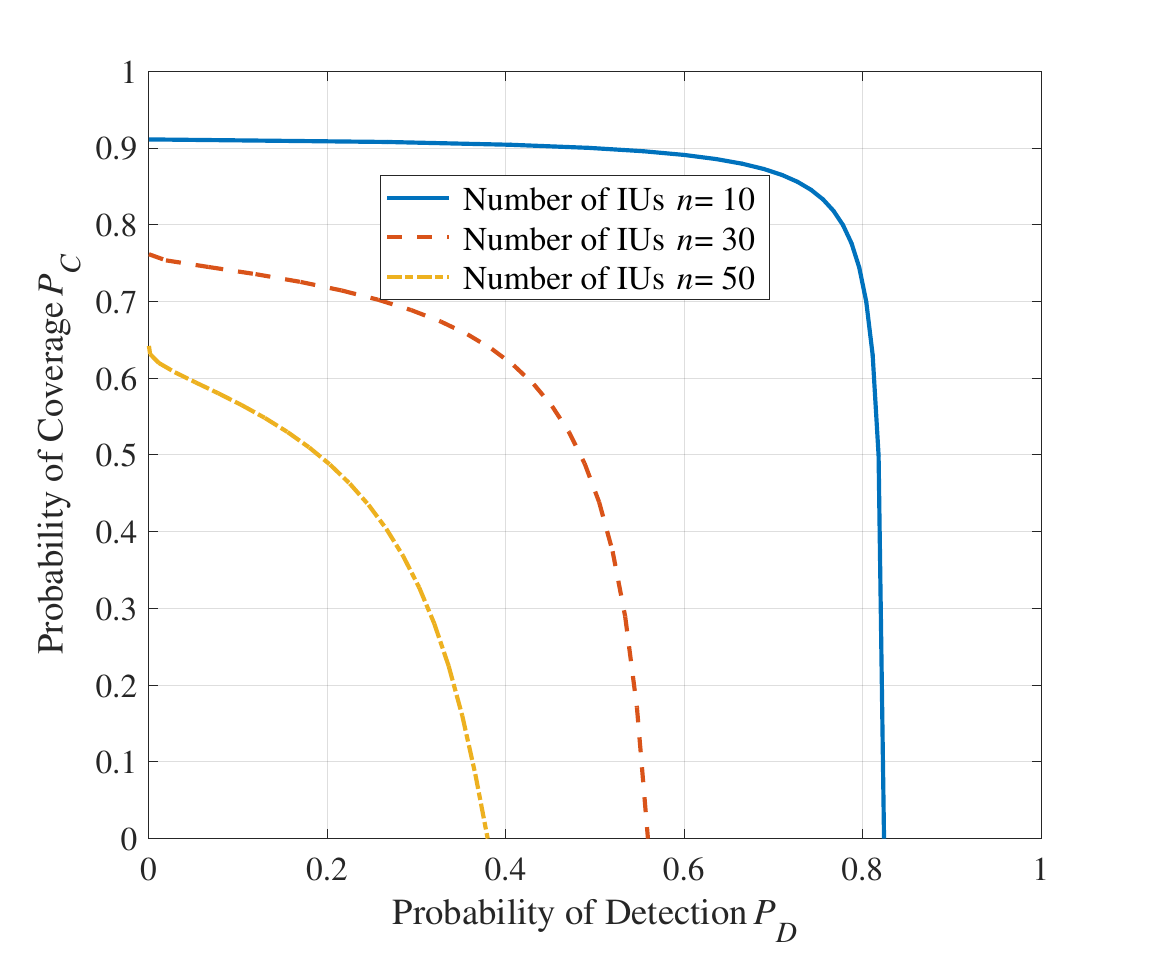}
 \caption{Performance trade-off of the ISAC network between the communication and sensing under power allocation.}
\label{fig:reultPD}
\end{minipage}
\begin{minipage}[t]{0.46\linewidth}
\centering
\includegraphics[width=1.83 in]{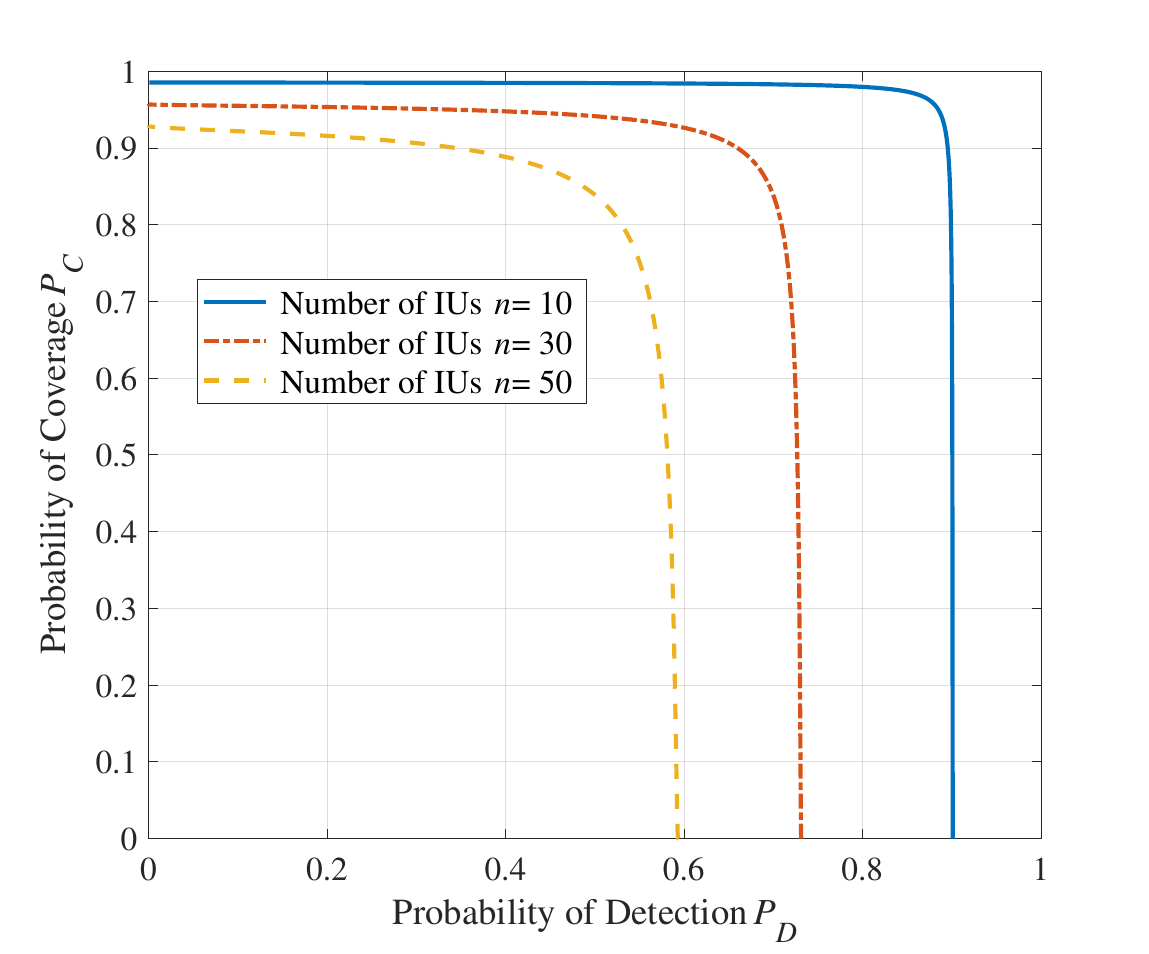}
\caption{Performance trade-off of the ISAC network between the communication and sensing under bandwidth allocation.}
\label{fig:reultPDB}
\end{minipage}
\end{figure}


As the main contribution of this paper, Fig.~\ref{fig:reultPD} demonstrates the fundamental performance trade-off of communication and sensing under power budget. Since the sum of the transmitting powers $p_c$ and $p_s$ are limited by the total power $p_t$, the performance of the communication system and sensing system could hardly be perfect simultaneously. However, achieving the best balance between the two performances is possible. In Fig.~\ref{fig:reultPD}, the points on each curve represent the maximal limitation value of the probability of coverage $P_C$ and probability of detection $P_D$ (under the condition that the total power $p_t$ is 30 dBm). When we increase the power $p_c$ to get the higher $P_C$, the power $p_s$ will decrease since the total power $p_t$ is limited, and the $P_D$ will decrease. Therefore, the curves provide the upper bound of $P_C$ and $P_D$ in a uniformed ISAC system, and the performance trade-off in a practical environment is likely to appear on the bottom left side of the curves. In this figure, different colors of the curves represent the different numbers of  IUs. Observe that when the numbers of IUs $n$ increase from 10 to 50, the upper boundary of $P_C$ and $P_D$ decreases rapidly. When the number of IUs $n$ is fixed, the increase of $P_D$ leads to the decrease of $P_C$ (the $p_s$ increases while the $p_c$ decreases).


When we allocate the bandwidth rather than the transmit power between the sensing and communication functions, we get the results of performance limits and performance trade-off of communication and sensing in the ISAC network shown by Fig.~\ref{fig:reultPDB}. Similar to the results of Fig.~\ref{fig:reultPD}, the results of Fig.~\ref{fig:reultPDB} illustrated the bandwidth resource for sensing and communication are also competitive. When the bandwidth $B_c$ increases, the $P_{C}$ increases while the $P_D$ decreases. This is because when the $B_c$ increase, the $B_s$ reduces due to the limited total bandwidth.
We also find that variations in bandwidth have a larger impact on the communication function than on the sensing function. When the bandwidth $B_c$  decreases from $ 20 \mathrm{~MHz}$ to 0 ( i.e., the bandwidth $B_s$ increases from 0 to $ 20 \mathrm{~MHz}$ ) and number of IUs $n=50$, the $P_{C}$ decreases from 0.928 to 0, while the $P_{D}$ increases from 0 to 0.596 (shown by the yellow cure).

\section{Conclusion}
In this letter, we proposed a stochastic geometry-based framework to analyze the sensing and communication performance trade-off of the distributed ISAC network.
Based on our analytical derivations, we provided a quantitative description of the performance limits and the performance trade-off between sensing and communication under the given transmit power and bandwidth budget. Our results are beneficial for the guidance of the power and bandwidth allocation between sensing and communication in a distributed ISAC network.

\bibliography{IEEEabrv,health}

\vfill

\end{document}